\documentclass[a4paper,11pt,english]{article}
\usepackage[utf8]{inputenc}
\usepackage{babel,amsmath,amssymb,amsthm,enumitem}
\usepackage[sort]{natbib}
\usepackage[colorlinks,allcolors=blue]{hyperref}
\usepackage[margin=1.25in]{geometry}
\usepackage[small,bf]{titlesec}
\titlelabel{\thetitle.\hspace{0.5em}}

\newtheorem{theorem}{Theorem}

\newtheorem{proposition}{Proposition}
\theoremstyle{definition}
\newtheorem{remark}{Remark}
\newtheorem{definition}{Definition}

\DeclareMathOperator{\E}{E}

\DeclareMathOperator{\I}{I}

\renewcommand{\tilde}{\widetilde}
\renewcommand{\epsilon}{\varepsilon}
\renewcommand{\P}{\mathrm{P}}
\newcommand{\F}{\mathcal{F}}
\newcommand{\FF}{\mathbb{F}}
\newcommand{\R}{\mathbb{R}}

\newcommand{\qc}[1]{\langle{#1}\rangle}
\newcommand{\qcc}[2]{\langle{#1,#2}\rangle}
\newcommand{\e}{\mathcal{E}}
\renewcommand*\bar[1]{\hbox{\vbox{\hrule height 0.5pt \kern0.5ex \hbox{\kern-0.1em \ensuremath{#1}\kern-0.1em}}}}

\title{Optimal growth strategies for a representative agent in~a~continuous-time asset market}
\author{Mikhail Zhitlukhin\thanks{Steklov Mathematical Institute of RAS, Moscow. Email: mikhailzh@mi-ras.ru.}}

\begin{document}
\maketitle

\begin{abstract}
We propose a multi-agent model of an asset market and study conditions that guarantee that the strategy of an individual agent cannot outperform the market. The model assumes a mean-field approximation of the market by considering an infinite number of infinitesimal agents who use the same strategy and another infinitesimal agent with a different strategy who tries to outperform the market. We show that the optimal strategy for the market agents is to split their investment budgets among the assets proportionally to their discounted expected relative dividend intensities. 
\end{abstract}

\section{Introduction}
The goal of this paper is to study growth optimal strategies in a multi-agent model of an asset market with endogenous asset prices. Recall that the standard notion of growth optimality in a market with exogenous prices can be formulated as follows: if $V_t$ denotes the wealth process of a growth optimal strategy and $W_t$ denotes the wealth process of any other strategy, then $W_t/V_t$ is a supermartingale. By exogenous asset prices we mean that they are specified as some stochastic processes which do not depend on the strategies used by the agents. In this paper, asset prices depend on agents' strategies. 

We consider a model which focuses on a mean-field representation of a market. There are $N$ dividend-paying assets and an infinite number of infinitesimal agents who use the same investment strategies and have the same wealth $V_t$. We call them representative agents. The strategy of these agents determines the asset prices through the short-run equilibrium of supply and demand. There are also other agents who may use other strategies, but they are small and do not influence the asset prices. If the wealth process of such an agent is $W_t$, we characterize the strategy of the representative agents such that $W_t/V_t$ is a supermartingale for any strategy of a small agent. 

This problem can be interpreted as the question to determine when it is not possible to ``beat'' the market. In the academic literature and among practitioners, it is generally believed  that an individual investor cannot outperform the market portfolio in a sufficiently long time frame. However, the market itself consists of individual investors and the market portfolio essentially represents the weighted strategy of all investors. Hence one can ask what the strategies of these investors should be, if they cannot be outperformed when considered as a whole. This paper provides a partial answer to this question.

There are two main results of the paper. The first result gives an explicit construction of a growth optimal strategy for representative agents and prove the uniqueness of this strategy. It turns out that representative agents must allocate their wealth among the assets proportionally to the total discounted expected relative dividend intensities of the assets. The second result shows that if an individual agent uses a strategy which is close to the growth optimal strategy, then he or she survives in the market in the sense that the ratio $W_t/V_t$ does not vanish as $t\to\infty$. If an individual agent uses a strategy which differs too much from the growth optimal strategy, then $W_t/V_t$ vanishes. We give a precise meaning of such closeness. Consequently, the growth optimal strategy, if used by the representative agents, forces other agents to use it or a close strategy as well, if they want to achieve investment performance not worse than that of the strategy which copies the market. 

Let us briefly mention how this paper is related to other results in the literature. The notion of growth optimality (for markets with exogenous prices) originated in the works of \cite{Kelly56} and \cite{Breiman61}. The growth optimality of the log-optimal strategy for a general discrete-time model was proved by \citet{AlgoetCover88}; a review of other results in discrete time can be found in, e.g., \citet[Chapter~16]{CoverThomas12} or \cite{HakanssonZiemba95}. For results in continuous time and a connection of growth optimal strategies (numeraire portfolios) with absence of arbitrage, see \cite{PlatenHeath06,KaratzasKardaras07}. 

The property of a market that it cannot be beaten by a specific agent is known to be related to the property of market diversity. The latter means that all the capital of the market does not become eventually accumulated in one asset, see, e.g., \cite{Kardaras08,FernholzKaratzas+05}. The results obtained in our paper bear resemblance to the results of \cite{Kardaras08}, but the model we consider is somewhat different, and, in particular, it allows to identify a single optimal strategy.

There is also a strand of papers in evolutionary finance which study the property of survival of investment strategies in markets where agents use arbitrary strategies and are not necessarily small. One of the main results in this direction consists of that the strategy that splits its investment budget among assets proportionally to their expected dividends survives in the market, see, e.g., \cite{AmirEvstigneev+11,AmirEvstigneev+20,EvstigneevHens+20}. In the framework of general equilibrium,  results on market selection of investment strategies were obtained, among others, by \cite{BlumeEasley06,Borovicka20,Sandroni00,Yan08}. A recent survey of literature in evolutionary finance can be found in \citet{Holtfort19}.

The paper is organized as follows. In Section~\ref{sec-model}, we formulate the model and introduce the notion of growth optimality of a strategy of a representative agent. In Section~\ref{sec-main}, we prove the two above-mentioned results. Section~\ref{sec-example} contains a simple example which illustrates the main results for a particular case of the general model where the optimal strategy and the notion of closeness have especially clear interpretation.

\section{Model}
\label{sec-model}
We consider a model of an asset market with an infinite number of agents, in which the actions of an individual agent do not affect market characteristics like asset prices or wealth distribution. The key assumption will be that the majority of agents use the same trading strategy and it is this strategy that determines the market characteristics. A possible interpretation of this assumption is that we consider a mean-filed approximation of a market consisting of diverse agents.

Let $(\Omega,\F, \P)$ be a probability space with a filtration $\FF=(\F_t)_{t\ge0}$ on which all random variables will be defined. We assume that it satisfies the usual assumptions (the filtration is right-continuous, the $\sigma$-algebra $\F$ is $\P$-complete and $\F_0$ contains all $\P$-null sets). We also assume that the filtration has the property that 
\begin{equation}
\label{cont-modification}
\text{any martingale on $(\Omega,\F,\FF, \P)$ has a continuous modification.}
\end{equation}
Recall that this is so if, for example, $\FF$ is the augmented filtration generated by a Brownian motion (single- or multi-dimensional). 

There are $N$ assets in the market which pay dividends with intensities $X_t^n$ per one share per unit of time, $n=1,\ldots,N$. The dividends are paid in some perishable good and must be consumed by the agents immediately; there is no possibility to store the good. The intensity processes $X_t^n$ are assumed to be non-negative continuous semimartingales satisfying the following conditions for any $t\ge 0$:
\begin{align}
\label{X-assumption-1}
&\bar X_t := \sum_{n=1}^N X_t^n > 0\ \text{a.s.},\\
\label{X-assumption-2}
&\P(\exists\, s\ge t : X_s^n>0 \mid \F_t) >0\ \text{for any}\ n=1,\ldots,N.
\end{align}
These are non-degeneracy conditions. The first one states that the total dividend intensity $\bar X_t$ is always  non-zero; the second one means that, for each asset, there is no moment of time $t$ after which it stops paying dividends with probability 1 conditionally on $\F_t$. 

The prices of the assets in the model are specified by strictly positive continuous semimartingales $S_t^n$ which will be defined endogenously based on the strategies of the agents.

An agent trading in this market is identified with his/her wealth process $W_t$, trading strategy $\lambda_t=(\lambda_t^1,\ldots,\lambda_t^N)$, and consumption rate $\rho>0$ (for simplicity, the consumption rate is assumed to be constant). The trading strategy specifies in what proportions this agent invests his/her capital in the assets. We assume $\lambda_t$ is a continuous semimartingale with components $\lambda_t^n\ge 0$ and $\sum_{n=1}^N \lambda_t^n = 1$ (short sales are not allowed). The evolution of wealth of an agent is described by the equation
\begin{equation}
\label{wealth}
d W_t = \sum_{n=1}^N \frac{\lambda_t^{n}W_t}{S_t^n} (d S_t^n + X_t^n dt) 
  - \rho W_t dt,
\qquad W_0 = w_0 > 0.
\end{equation}

Our key assumption will be that, informally speaking, the market consists of an infinite number of identical infinitesimal agents, further called \textit{representative agents}, who use the same strategy and have the same wealth process. These agent occupy the total proportion 1 in the set of agents and their actions define the evolution of asset prices. At the same time, there may be other infinitesimal agents (of total proportion 0), further called \emph{small agents}, whose actions do not have any effect on the asset prices. We will assume that the consumption rate $\rho$ is the same for all agents, since our goal will be to find a strategy which generates more wealth; otherwise a strategy with a smaller consumption rate will have an advantage.

Let $V_t$ and $\mu_t=(\mu_t^1,\ldots,\mu_t^N)$ be, respectively, the wealth process and the strategy of each representative agent. Then we postulate that the asset prices $S_t^n$ are defined by the relation
\begin{equation}
\label{prices}
S_t^n = s^n\mu_t^n V_t,
\end{equation}
where $s^n>0$ are some constants. Equation \eqref{prices} has the following interpretation. Assume that the supply volume (the number of shares) of asset $n$ is equal to $s^n M$, where $M\to\infty$ is the number of representative agents trading in this market. Then in order to reach the equilibrium of supply and demand at each moment of time, it must hold that
\begin{equation}
\label{supply-demand}
s^n M = M\frac{\mu_t^n V_t}{S_t^n},
\end{equation}
where the right-hand side represents the demand of the agents for asset $n$, if we assume that  agents who use strategies different from $\mu_t$ have a vanishing share in the market in the limit $M\to\infty$. Then \eqref{supply-demand} implies \eqref{prices}.

Certainly, the above limit argument lacks mathematical rigor, but we just \emph{assume} that the model is defined by equations \eqref{wealth}--\eqref{prices}. In order to avoid the case of zero asset prices, which would make equation~\eqref{wealth} ill-defined, let us call a strategy $\mu_t$ of a representative agent \emph{admissible} if
\[
\mu_t^n >0\ \text{a.s.\ for all $t\ge0$, $n=1,\ldots,N$.}
\]
We will consider only admissible strategies $\mu_t$.

Note that if we plug $\lambda_t = \mu_t$ and $W_t=V_t$ in \eqref{wealth}, we get
\[
d V_t = \sum_{n=1}^N \left(d(\mu_t^n V_t) + \frac{X_t^n}{s^n} dt\right) 
  - \rho V_t dt 
= d V_t + \sum_{n=1}^n \frac{X_t^n}{s^n} dt - \rho V_t dt,
\]
where we used that $\sum_{n=1}^N \mu_t^n =  1$. Hence the capital of a representative agent is simply 
\[
V_t = \frac1\rho \sum_{n=1}^N \frac{X_t^n}{s^n}.
\]
In what follows, without loss of generality, let us assume that $s^n=1$ for all $n=1,\ldots,N$, so we have
\begin{equation}
\label{total-wealth}
V_t = \frac{\bar X_t}{\rho},
\end{equation}
where $\bar X_t = \sum_{n=1}^N X_t^n$.

Now, in view of \eqref{total-wealth}, equation \eqref{wealth} for the wealth of a small agent becomes
\[
d W_t = \sum_{n=1}^N \frac{\lambda_t^{n}W_t}{\mu_t^n \bar X_t} 
  (d (\mu_t^n \bar X_t) + \rho X_t^n dt) - \rho W_t dt, 
\qquad W_0=w_0>0.
\]
We will be interested in the ratio $W_t/V_t$ of the wealth of a small agent to the wealth of a representative agent. 

\begin{definition}
We say that a (admissible) strategy $\mu_t$ of a representative agent is \emph{growth optimal} if for any strategy $\lambda_t$ of a small agent
\[
\frac{W_t}{V_t}\ \text{is a supermartingale}.
\]
\end{definition}

This definition expresses the idea that a small agent achieves the best performance of his/her investment if he/she uses the same strategy as the representative agents: in other words, one cannot beat the market. This is analogous to the well-known notion of a growth optimal strategy (or a numeraire portfolio) in markets with exogenous asset prices, see, e.g., \cite{PlatenHeath06,KaratzasKardaras07}.

A simple corollary from the growth optimality of the representative agent's strategy is that a small agent cannot asymptotically outperform it with probability 1 in the sense of the following proposition.

\begin{proposition}
If the representative agents use a growth optimal strategy, than for any strategy of a small agent there exists a (finite-valued) random variable $C$ such that 
\[
\frac{W_t}{V_t} \le C\ \text{a.s.}
\]
\end{proposition}
\begin{proof}
This result immediately follows from that a non-negative supermartingale has an a.s.-finite limit.
\end{proof}

This proposition implies that a market of representative agents cannot be ``invaded'' by a small agent which would eventually, as $t\to\infty$, become a not-so-small agent in the sense of holding a non-infinitesimal share of market wealth. This interpretation is analogous to the notion of an evolutionary stable strategy of \cite{MaynardSmithPrice73}.

In Theorem~\ref{th-2} below, we show that a small agent does not lose asymptotically to a representative agent, i.e.\ $\inf_{t\ge 0} W_t/V_t >0$, only if he/she uses a strategy which is in a certain sense close to $\mu_t$.

\section{Main results}
\label{sec-main}
Our first main result finds a growth optimal strategy for a representative agent in a closed form. Let $R_t^n$ denote the relative dividend intensities
\[
R_t^n = \frac{X_t^n}{\bar X_t}.
\]
Define the strategy
\begin{equation}
\label{optimal}
\mu_t^n = \int_t^\infty \rho e^{\rho (t-s)} \E(R_s^n\mid \F_t) ds.
\end{equation}
Note that $\mu_t$ is admissible in view of assumptions \eqref{X-assumption-1}--\eqref{X-assumption-2}. We are going to prove that the strategy $\mu_t$ is growth optimal, and it is a unique growth optimal strategy with respect to the measure $(\P\otimes\mathrm{Leb})$, where $\mathrm{Leb}$ stands for the Lebesgue measure on $\R_+$.  The idea how to find this strategy will be explained in Remark~\ref{rem-1} after the proof.

\begin{theorem}
\label{th-1}
The strategy $\mu_t$ defined in \eqref{optimal} is a growth optimal strategy for a representative agent. Moreover, $\mu_t$ is a $(\P\otimes\mathrm{Leb})$-unique growth optimal strategy.
\end{theorem}

\begin{proof}
Without loss of generality, let us assume that $W_0=V_0=1$. By Ito's formula, one can check that the ratio process $W_t/V_t$ can be represented in the form
\[
\frac{W_t}{V_t} = \e(Z)_t,
\]
where $\e(Z)$ is the stochastic exponent of the process $Z_t$ with the stochastic differential
\begin{equation}
\label{Zt}
d Z_t = \sum_{n=1}^N \frac{\lambda_t^n}{\mu_t^n} d \mu_t^n 
  + \rho\left(\sum_{n=1}^N \frac{\lambda_t^n}{\mu_t^n}R_t^n - 1\right)dt,
\qquad Z_0 = 0.
\end{equation}
(Recall that the stochastic exponent of a continuous semimartingale $Z_t$ is the unique process $\e(Z)_t$ which solves the equation $d\e(Z)_t = \e(Z)_t d Z_t$, $\e(Z)_0 = 1$; see, e.g., \citet[Ch.~II.8]{JacodShiryaev02} for details.)

Using the explicit form of the process $\mu_t$, we find
\begin{equation}
\label{mu-sde}
d \mu_t^n = \rho(\mu_t^n - R_t^n) dt + e^{\rho t}d L_t^n, 
\end{equation}
where the processes $L_t^n$ are given by
\begin{equation}
\label{L}
L_t^n = \E\left(\int_0^\infty \rho e^{-\rho s}R_s^n ds \;\bigg|\; \F_t \right).
\end{equation}
Note that $L_t^n$ are martingales, and we can assume they are continuous in view of assumption \eqref{cont-modification}. From this, one can see that $Z_t$ can be represented in the form
\begin{equation}
\label{Z-martingale}
d Z_t = e^{\rho t}\sum_{n=1}^N \frac{\lambda_t^n}{\mu_t^n} d L_t^n.
\end{equation}
Consequently, $Z_t$ is a local martingale, and hence $W_t/V_t$ is a local martingale. Since the latter process is non-negative, it is a supermartingale. This proves that $\mu_t$ is a growth optimal strategy.

To prove the uniqueness, consider an arbitrary admissible strategy $\tilde\mu_t$ of the representative agent. Since it is a bounded continuous semimartingale, it can be represented in the form
\[
d\tilde\mu_t^n = a_t^n dt + d U_t^n + e^{\rho t}d \tilde L_t^n,
\]
where $a_t^n$, $n=1,\ldots,N$, are some locally integrable processes, $U_t^n$ are continuous processes of bounded variation which are a.s.\ singular with respect to the Lebesgue measure, and $\tilde L_t^n$ are martingales. We can assume $\tilde L_0^n=0$ and $U_0^n = 0$. Then for the corresponding process $Z_t$ from \eqref{Zt}, we have
\[
d Z_t = \sum_{n=1}^N \left(
  \frac{\lambda_t^n}{\tilde \mu_t^n}(a_t^n + \rho (R_t^n-\tilde \mu_t^n)) dt 
  + \frac{\lambda_t^n}{\tilde \mu_t^n} d U_t^n 
  + e^{\rho t}\frac{\lambda_t^n}{\tilde \mu_t^n} d \tilde L_t^n
\right).
\]
If it does not hold that $a_t^n = \rho(\tilde \mu_t^n - R_t^n)$ $(\P\otimes\mathrm{Leb})$-a.e.\ for all $n$ and $U_t^n = 0$ a.s.\ for all $n$, then it is possible to find a strategy $\lambda_t$ such that $Z_t$ is a strict submartingale (i.e.\ $\E (Z_t \mid \F_s) > Z_s$ for some $t,s$). Hence, if $\tilde\mu_t$ is a growth optimal strategy for a representative agent, it must hold that
\[
d \tilde \mu_t^n = \rho(\tilde \mu_t^n - R_t^n) dt + e^{\rho t}d \tilde L_t^n.
\]
From this equation and \eqref{mu-sde}, it follows that 
\[
d (\mu_t^n - \tilde \mu_t^n) = \rho(\mu_t^n - \tilde\mu_t^n) dt + d M_t^n,
\]
where $dM_t^n = e^{\rho t} d(L_t^n - \tilde L_t^n)$. Consequently, for all $n$ and $s\ge t$ we have
\[
\E(\mu_s^n - \tilde\mu_s^n \mid \F_t)\I(\mu_t^n > \tilde\mu_t^n) 
= e^{\rho(s-t)} (\mu_t^n - \tilde \mu_t^n)\I(\mu_t^n > \tilde\mu_t^n).
\]
The right-hand side tends to $+\infty$ as $s\to\infty$ on the set $\{\mu_t^n>\tilde\mu_t^n\}$, hence this set has zero probability because the left-hand side is bounded by 1. In the same way we prove that $\P(\mu_t^n < \tilde\mu_t^n) = 0$. Thus $\tilde\mu_t^n = \mu_t^n$ a.s.
\end{proof}

\begin{remark}
\label{rem-1}
From the proof of Theorem~\ref{th-1}, one can see the idea to find the optimal strategy $\mu_t$: it must solve stochastic differential equation \eqref{mu-sde}, in which the processes $\mu_t^n$ and $L_t^n$ are unknown, and must also satisfy the conditions $\mu_t^n\ge 0$ and $\sum_n \mu_t^n =  1$.

In order to solve this equation, let us additionally assume that the underlying probability space has the property that any continuous martingale can be represented as a stochastic integral with respect to some $d$-dimensional martingale $M_t$, i.e. $e^{\rho t}dL_t^n = \sum_{i=1}^d b_t^{ni} d M_t^i$. Then we get the following system of equations with unknown processes $\mu_t^n$ and $b_t^{in}$:
\[
d \mu_t^n = \rho(\mu_t^n - R_t^n) dt + \sum_{i=1}^db_t^{ni} d M_t^i, 
\qquad n=1,\ldots,N.
\]
This system of equations resembles a linear BSDE, except that instead of a terminal condition $\mu_T = \xi$, we have the constraints $\mu_t^n\ge 0$, $\sum_n \mu_t^n = 1$. However, if we assume that the solution $\mu_t$ exists, then for any $T>0$ we can indeed consider it as a BSDE and from the well-known formula for the solution of a linear BSDE find that
\[
\mu_t^n = e^{\rho t} \E\left( e^{-\rho T}\mu_T^n 
+ \int_t^T \rho e^{-\rho s} R_s^n ds \;\bigg|\;\F_t\right).
\]
Since $\mu_T^n \in[0,1]$, we can let $T\to\infty$ and obtain that $\mu_t^n$ is defined by \eqref{optimal}.
\end{remark}

Our second result gives a necessary and sufficient condition for a small agent to \emph{survive} in the market, in the sense that his/her wealth does not become asymptotically vanishing compared to the wealth of a representative agent. One can interpret this condition as that the strategy $\lambda_t$ must be close to $\mu_t$. Its meaning will become especially transparent in a particular example considered in the next section.

\begin{theorem}
\label{th-2}
Suppose that the representative agents use the growth optimal strategy $\mu_t$ and an individual agent uses some strategy $\lambda_t$. Define the process
\[
G_t = \sum_{i,j=1}^N \int_0^t e^{2\rho s} 
\frac{\lambda_s^i\lambda_s^j}{\mu_s^i\mu_s^j} d \qcc{L^i}{L^j}_s,
\]
where $L^i_t$ are the continuous martingales defined in \eqref{L} and the angle brackets denote the quadratic covariation. 

Then $G_t$ is non-decreasing and 
\begin{align}
\label{ratio-limit-1}
&\lim_{t\to\infty} \frac{W_t}{V_t} = 0
  \ \text{a.s.\ on the set}\ \{G_\infty = \infty\},\\
\label{ratio-limit-2}
&\lim_{t\to\infty} \frac{W_t}{V_t} > 0
  \ \text{a.s.\ on the set}\ \{G_\infty < \infty\}.
\end{align}
\end{theorem}

\begin{proof}
By Ito's formula, one can find that
\[
\ln\frac{W_t}{V_t} = \ln\frac{W_0}{V_0} + Z_t - \frac12 \qc Z_t,
\]
where $Z_t$ is the local martingale from the proof of Theorem~\ref{th-1}. From~\eqref{Z-martingale}, we see that $\qc Z_t = G_t$. 
Then on the set $\{G_\infty = \infty\}$, we have $\lim_{t\to\infty} M_t /G_t =0$  by the SLLN for martingales, which proves \eqref{ratio-limit-1}. On the set $\{G_\infty < \infty\}$, there exists the finite limit $M_\infty=  \lim_{t\to\infty} M_t$, which proves \eqref{ratio-limit-2}. (Regarding the SLLN and convergence of martingales, see, e.g., \citet[Ch.~2.6]{LiptserShiryaev89}.)
\end{proof}

\section{Example}
\label{sec-example}
The following example illustrates the main results of the paper. Let the filtration $\FF$ be generated by a $K$-dimensional standard Brownian motion $B_t=(B_t^1,\ldots,B_t^K)$ and the relative dividend intensity process $R_t$ be a martingale of the form
\begin{equation}
\label{example-R}
d R_t = \sigma_t d B_t,
\end{equation}
where $\sigma_t=(\sigma_t^{nk}) \in \R^{N\times K}$ is a matrix-valued process integrable with respect to $B_t$. A simple particular model of the dividend intensity process $X_t$ that leads to equation \eqref{example-R}, in the case $K=1$, $N=2$ is given by
\begin{align}
\label{example-1d-1}
&d X_t^1 = \sigma X_t^1(1-X_t^1) d B_t, \qquad X_0^1 = x_0^1 \in (0,1),\\
\label{example-1d-2}
&X_t^2 = 1- X_t^1,
\end{align}
where $\sigma>0$ is a parameter. It is easy to see that $X_t^1$ always stays in the interval $(0,1)$, hence the model is well-defined. Obviously, in this case $R_t^i = X_t^i$, $i=1,2$.

As follows from Theorem~\ref{th-1}, in model \eqref{example-R}, the optimal strategy prescribes to allocate the investment budget among the assets  proportionally to the relative dividend intensities:
\[
\mu_t^n = R_t^n.
\]
From \eqref{mu-sde}, we have $dL_t^i = e^{-\rho t} d\mu_t^i$. Define $\nu_t = (\nu^1_t,\ldots,\nu^N_t)$ with $\nu_t^n = \lambda_t^n/ R_t^n$. Then the process $G_t$ from Theorem~\ref{th-2} becomes
\[
G_t = \int_0^t \|\nu_s\sigma_s\|^2 ds.
\]
Consequently, by Theorem~\ref{th-2}, we have 
\[
\lim_{t\to\infty} \frac{W_t}{V_t} = 0
\ \text{a.s.\ if and only if}\ \int_0^\infty \|\nu_t \sigma_t\|^2 dt =\infty.
\]
For the model defined by \eqref{example-1d-1}--\eqref{example-1d-2}, it is easy to see that
\[
G_t = \sigma^2\int_0^t (\lambda_s^1 - R_s^1)^2 ds.
\]
Therefore, the convergence $W_t/V_t\to0$ takes place if and only if $\int_0^\infty (\lambda_t^1 - \mu_t^1)^2 dt = \infty$.

\small 
\setlength{\bibsep}{0.2em plus 0.3em}
\bibliographystyle{apalike}
\bibliography{growth-optimal}

\end{document}